\PassOptionsToPackage{unicode}{hyperref}
\PassOptionsToPackage{hyphens}{url}
\PassOptionsToPackage{dvipsnames,svgnames,x11names}{xcolor}
\documentclass[12pt]{article}

\newcommand{\bc}{\textbf{c}}
\newcommand{\blam}{\boldsymbol \lambda}

\newcommand{\by}{\textbf{y}}
\newcommand{\bth}{\boldsymbol \theta}
\newcommand{\brho}{\boldsymbol \rho}
\newcommand{\bnu}{\boldsymbol \nu}

\newcommand{\balpha}{\boldsymbol \alpha}
\newcommand{\bbeta}{\boldsymbol \beta}

\newcommand{\cinew}{c_i^{*}}

\usepackage{amsthm}
\newtheorem{proposition}{Proposition}
\usepackage{amsmath,amssymb}
\usepackage{booktabs, multirow, xcolor}
\usepackage{iftex}
\ifPDFTeX
  \usepackage[T1]{fontenc}
  \usepackage[utf8]{inputenc}
  \usepackage{textcomp} 
\else 
  \usepackage{unicode-math}
  \defaultfontfeatures{Scale=MatchLowercase}
  \defaultfontfeatures[\rmfamily]{Ligatures=TeX,Scale=1}
\fi
\usepackage{lmodern}
\ifPDFTeX\else  
\fi
\IfFileExists{upquote.sty}{\usepackage{upquote}}{}
\IfFileExists{microtype.sty}{
  \usepackage[]{microtype}
  \UseMicrotypeSet[protrusion]{basicmath} 
}{}
\makeatletter
\@ifundefined{KOMAClassName}{
  \IfFileExists{parskip.sty}{%
    \usepackage{parskip}
  }{
    \setlength{\parindent}{0pt}
    \setlength{\parskip}{6pt plus 2pt minus 1pt}}
}{
  \KOMAoptions{parskip=half}}
\makeatother
\usepackage{xcolor}
\makeatletter
\ifx\paragraph\undefined\else
  \let\oldparagraph\paragraph
  \renewcommand{\paragraph}{
    \@ifstar
      \xxxParagraphStar
      \xxxParagraphNoStar
  }
  \newcommand{\xxxParagraphStar}[1]{\oldparagraph*{#1}\mbox{}}
  \newcommand{\xxxParagraphNoStar}[1]{\oldparagraph{#1}\mbox{}}
\fi
\ifx\subparagraph\undefined\else
  \let\oldsubparagraph\subparagraph
  \renewcommand{\subparagraph}{
    \@ifstar
      \xxxSubParagraphStar
      \xxxSubParagraphNoStar
  }
  \newcommand{\xxxSubParagraphStar}[1]{\oldsubparagraph*{#1}\mbox{}}
  \newcommand{\xxxSubParagraphNoStar}[1]{\oldsubparagraph{#1}\mbox{}}
\fi
\makeatother

\usepackage{longtable,booktabs,array}
\usepackage{calc} 
\usepackage{etoolbox}
\makeatletter
\patchcmd\longtable{\par}{\if@noskipsec\mbox{}\fi\par}{}{}
\makeatother
\IfFileExists{footnotehyper.sty}{\usepackage{footnotehyper}}{\usepackage{footnote}}
\makesavenoteenv{longtable}
\usepackage{graphicx}
\makeatletter
\def\maxwidth{\ifdim\Gin@nat@width>\linewidth\linewidth\else\Gin@nat@width\fi}
\def\maxheight{\ifdim\Gin@nat@height>\textheight\textheight\else\Gin@nat@height\fi}
\makeatother
\setkeys{Gin}{width=\maxwidth,height=\maxheight,keepaspectratio}
\makeatletter
\def\fps@figure{htbp}
\makeatother

\makeatletter
\@ifpackageloaded{caption}{}{\usepackage{caption}}
\AtBeginDocument{%
\ifdefined\contentsname
  \renewcommand*\contentsname{Table of contents}
\else
  \newcommand\contentsname{Table of contents}
\fi
\ifdefined\listfigurename
  \renewcommand*\listfigurename{List of Figures}
\else
  \newcommand\listfigurename{List of Figures}
\fi
\ifdefined\listtablename
  \renewcommand*\listtablename{List of Tables}
\else
  \newcommand\listtablename{List of Tables}
\fi
\ifdefined\figurename
  \renewcommand*\figurename{Figure}
\else
  \newcommand\figurename{Figure}
\fi
\ifdefined\tablename
  \renewcommand*\tablename{Table}
\else
  \newcommand\tablename{Table}
\fi
}
\@ifpackageloaded{float}{}{\usepackage{float}}
\floatstyle{ruled}
\@ifundefined{c@chapter}{\newfloat{codelisting}{h}{lop}}{\newfloat{codelisting}{h}{lop}[chapter]}
\floatname{codelisting}{Listing}

\makeatother
\makeatletter
\@ifpackageloaded{caption}{}{\usepackage{caption}}
\@ifpackageloaded{subcaption}{}{\usepackage{subcaption}}
\makeatother

\ifLuaTeX
  \usepackage{selnolig}  
\fi
\usepackage[]{natbib}
\usepackage{bookmark}

\IfFileExists{xurl.sty}{\usepackage{xurl}}{} 
\hypersetup{
  pdftitle={Bayesian inference for non-conjugate distance dependent Chinese restaurant process models},
  pdfauthor={Joseph Marsh; Theodore Kypraios; Rowland G Seymour},
  pdfkeywords={bayesian non-parameteric clustering},
  colorlinks=true,
  linkcolor={blue},
  filecolor={Maroon},
  citecolor={Blue},
  urlcolor={Blue},
  pdfcreator={LaTeX via pandoc}}

\newcommand{\anon}{1}

\begin{document}

\def\spacingset#1{\renewcommand{\baselinestretch}%
{#1}\small\normalsize} \spacingset{1}


\if1\anon
{
  \title{\bf Bayesian Inference for Non-Conjugate Distance Dependent Chinese Restaurant Process Models}
  \author{Joseph Marsh\\
    School of Mathematics, University of Birmingham\\
     \\
    Theodore Kypraios \\
    School of Mathematical Sciences, University of Nottingham \\
     \\
    Rowland G. Seymour \\
    School of Mathematics, University of Birmingham}
  \maketitle
} \fi

\if0\anon
{
  \bigskip
  \bigskip
  \bigskip
  \begin{center}
    {\LARGE\bf Bayesian Inference for Non-Conjugate Distance Dependent Chinese Restaurant Process Models}
\end{center}
  \medskip
} \fi

\bigskip
\begin{abstract}
The distance dependent Chinese Restaurant Process (ddCRP) provides a flexible prior distribution for clustering observations, incorporating covariate information through pairwise distances and accommodating a rich variety of cluster structures. When cluster parameters are conjugate to the likelihood, Bayesian inference is straightforward. In the non-conjugate setting, however, inference becomes substantially more challenging due to the trans-dimensional parameter spaces that arise as cluster assignments change. We develop a reversible jump Markov chain Monte Carlo (RJMCMC) framework to address this challenge, targeting the dimension-changing nature of cluster parameter vectors when observation assignments are updated. We introduce and compare several proposal strategies for birth and death moves, including prior-based, independence, and data-driven moment-matching proposals that target regions of high posterior density. For fixed-dimensional moves, we propose a posterior resampling strategy that improves acceptance rates while maintaining computational efficiency. Through a simulation study and an application to Old Faithful eruption durations, we demonstrate moment-matched proposals offer a principled, data-driven alternative to prior-based proposals. The resulting methodology provides a general RJMCMC framework for ddCRP models with non-conjugate likelihoods, demonstrated here on both discrete and continuous observation models.
\end{abstract}

\noindent%
{\it Keywords:} Bayesian nonparametric clustering; random partitions; reversible jump MCMC
\vfill

\newpage
\spacingset{1} 

\section{Introduction}

Many clustering problems arise in settings where observations carry spatial, temporal, or network structure, and where substantive knowledge suggests that nearby observations (in the broad sense) are more likely to belong to the same group. The distance-dependent Chinese Restaurant Process (ddCRP) of \citet{blei2011} provides a natural Bayesian nonparametric framework for such problems. Rather than constructing dependent random measures via the stick-breaking representation \citep{sethuraman1994, maceachern1999}, the ddCRP operates directly at the level of observation-to-observation assignments: each observation $i$ is linked to another observation $j$ with probability proportional to a decay function $f(d_{ij})$ of their pairwise distance, or to itself with probability proportional to a concentration parameter $\alpha$. The partition is then determined by the connected components of the resulting directed graph. The ddCRP reduces to the standard Chinese Restaurant Process (CRP) \citep{aldous1985, pitman1995} when all distances are equal, and has been extended to incorporate spatial contiguity constraints \citep{ghosh2011} and auxiliary covariate information \citep{li2013, li2016}. \citet{dahl2017} proposed a broader class of random partition distributions indexed by pairwise information in a related spirit.

When conjugacy holds between the cluster likelihood and its prior, inference for ddCRP models is straightforward: cluster-specific parameters can be analytically marginalised, yielding tractable full conditionals for a sequential Gibbs sampler that updates one observation assignment at a time \citep{blei2011}. The analogous conjugate machinery is well established for the exchangeable CRP, where marginal samplers based on the P\'{o}lya urn scheme \citep{blackwell1973, escobar1995, neal2000} and blocked Gibbs samplers based on truncated stick-breaking \citep{ishwaran2001} are standard tools. In many practically important settings, however, conjugacy does not hold: mixtures of gamma distributions with unknown shape parameters, mixtures of skew-normal distributions, and many other flexible component models all give rise to non-conjugate cluster likelihoods. For the exchangeable CRP, \citet{neal2000} addressed this by proposing auxiliary-variable algorithms and Metropolis--Hastings updates for allocation variables, and \citet{jain2004, jain2007} developed split-merge samplers that substantially improve mixing by simultaneously reallocating subsets of observations between clusters. These methods are designed specifically for the exchangeable CRP and do not transfer directly to the ddCRP, where reassigning a single observation link can create or destroy clusters and hence change the dimension of the parameter space, and where any reassignment must maintain the consistency of the underlying directed graph. At present, to our knowledge no general-purpose inference algorithm exists for non-conjugate ddCRP models.

The current state of inference methodology for Bayesian nonparametric mixture models can therefore be summarised as follows. For exchangeable CRP models, both conjugate and non-conjugate inference are well developed. For the ddCRP, inference has so far been restricted to the conjugate setting. The Gibbs sampler of \citet{blei2011} requires conjugacy by construction, since the full conditional for each assignment depends on analytically marginalising the cluster-specific parameters; the variational inference algorithm of \citet{bartunov2014}, developed for the sequential ddCRP, operates under the same restriction. The non-conjugate ddCRP remains an open problem, and this is the gap we address.

In this paper we develop a reversible jump Markov chain Monte Carlo (RJMCMC) framework \citep{green1995} for posterior inference in ddCRP models with non-conjugate cluster likelihoods. The key challenge is that each observation reassignment may alter the number of clusters and hence the dimension of the parameter space. We address this by formulating the assignment updates as trans-dimensional moves within the RJMCMC framework of \citet{green1995}, with dimension-matching achieved through explicit proposal distributions for the parameters of newly created clusters. We introduce a range of proposal strategies for the birth and death moves that arise when an assignment change creates or destroys a cluster, including prior-based proposals, independence proposals, and data-driven moment-matching proposals that target statistics of the observations whose cluster membership changes. We additionally develop strategies for fixed-dimensional moves in which the partition changes but the number of clusters does not, and we derive exact data-augmented updates for the ddCRP concentration hyperparameter. The framework is validated by comparison with the marginalised Gibbs sampler of \citet{blei2011} in a conjugate Poisson mixture setting, and applied to eruption durations from the Old Faithful geyser dataset using a non-conjugate gamma likelihood.

This paper is organised as follows. Section~2 reviews the ddCRP and its conjugate inference algorithm before presenting the RJMCMC framework, proposal strategies, and hyperparameter inference schemes. Section~3 presents the simulation study and real-data application. A brief discussion closes the paper.

\section{Materials and Methods}

We begin by reviewing the ddCRP and establishing notation for the general model class. We then describe exact inference for conjugate models via Gibbs sampling, before developing the RJMCMC framework for non-conjugate settings, including the proposal strategies and optional within-algorithm hyperparameter inference.

\subsection{The distance dependent Chinese Restaurant Process}

The distance dependent Chinese Restaurant Process (ddCRP) \citep{blei2011} defines a prior distribution over partitions of $n$ observations by means of customer-to-customer links that are informed by pairwise distances. Unlike the Chinese Restaurant Process (CRP), which generates exchangeable partitions, the ddCRP incorporates covariate information through a distance structure, encouraging nearby observations to cluster together.

Let $\by = (y_1, \ldots, y_n)$ denote a collection of $n$ observations and let $\mathbf{D} = (d_{ij})_{i,j=1}^n$ be a symmetric matrix of pairwise distances between the observations, where $d_{ij} \geq 0$ and $d_{ii} = 0$. We define a customer assignment vector $\mathbf{c} = (c_1, \ldots, c_n)$ where each $c_i \in \{1, \ldots, n\}$ represents the customer that customer $i$ links to. The links are drawn independently according to
$$
\Pr(c_i = j \mid \mathbf{D}, \alpha, f) \propto
\begin{cases}
f(d_{ij}) & \text{if } j \neq i, \\
\alpha & \text{if } j = i,
\end{cases}
$$
where $f$ is a non-increasing decay function and $\alpha > 0$ is a concentration parameter governing the propensity for self-links. A common choice is the exponential decay function $f(d) = \exp(-s \, d)$ with scale parameter $s > 0$, though any non-negative, non-increasing function is admissible.

Since the links are conditionally independent given the pairwise distances $\mathbf{D}$, the joint prior over the assignment vector factorises as
$$
\pi(\mathbf{c} \mid \mathbf{D}, \alpha, f) = \prod_{i=1}^n p(c_i \mid \mathbf{D}, \alpha, f).
$$
For notational convenience we denote $\eta=(\textbf{D}, \alpha, f, \bnu)$ to be the vector of hyperparameters associated with the ddCRP where $\bnu$ are any parameters necessary for the decay function $f$. The assignment vector $\mathbf{c}$ defines a directed graph on $\{1, \ldots, n\}$ in which each node $i$ has exactly one outgoing edge to $c_i$. Similar to the notation used in \cite{blei2011}, we define $z(\bc)$ to be the vector of table assignments induced by the customer assignments, obtained by identifying the connected components of this graph. Furthermore, define $|z(\bc)|$ to be the number of clusters, $z^k(\bc)$ to be the set of indices associated with cluster $k=1,...,|z(\bc)|$, and $z_i(\bc) \in \{1, \ldots, |z(\bc)|\}$ to be the cluster index of observation $i$, such that $i \in z^{z_i(\bc)}(\bc)$.

In the original CRP and ddCRP literature, data points are referred to as \textit{customers} and partition groups as \textit{tables}; these terms are in general used interchangeably. For clarity, we use \textit{observation} in place of customer and \textit{cluster} in place of table throughout the remainder of this article.

Conditional on the partition, observations $\by_{z^k(\bc)}$ in cluster $k$ are assumed to be independent and identically distributed with density $f(\by_{z^k(\bc)} \mid \bth, \brho_k)$ with respect to an appropriate base measure. The parameters are $\brho_k$, the vector of cluster specific parameters and $\bth$, the vector of global parameters shared across all clusters, which may be empty. Since observations are independent of one another, the density of the data is a product of cluster-specific contributions and is explicitly given by 
$$
f(\by \mid \bth, \brho, \bc) = \prod_{k=1}^{|z(\bc)|} f(\by_{z^k(\bc)} \mid \bth, \brho_k),
$$
where $\brho = (\brho_1, ..., \brho_{|z(\bc)|})$ is the collection of all cluster specific parameters.

By assigning independent priors $\pi(\bth)$ and $\pi(\brho_k \mid \bc)$ for $k=1,\ldots,|z(\bc)|$, the joint distribution factorises as
$$
\pi(\by, \bth, \brho, \bc, \eta) = f(\by \mid \bth, \brho, \bc)  \pi(\bth) \pi(\brho \mid \bc) \pi(\bc \mid \eta).
\label{eq:joint_distribution}
$$
Since the marginal likelihood is intractable, inference proceeds by targeting the unnormalised posterior, which by Bayes' theorem is proportional to the joint distribution above. Our primary interest is in the posterior over assignments $\bc$ and parameters $(\bth, \brho)$.

\subsection{Inference for conjugate models}

If the prior distribution on the cluster specific parameters is conjugate, we may marginalise out those parameters to obtain the marginal cluster likelihood. Explicitly, this target distribution is given by
\begin{align*}
    \pi(\bth, \bc \mid \by, \eta) &= \int \pi(\bth, \brho, \bc \mid \by, \eta) d \brho \\
    &= \pi(\bth) \pi(\bc \mid \eta) \prod_{k=1}^{|z(\bc)|} \int  f(\by_{z^k(\bc)} \mid \bth, \brho_k, \bc) \pi(\brho_k \mid \bc)  d\brho_k \\
    &= \pi(\bth) \pi(\bc \mid \eta) \prod_{k=1}^{|z(\bc)|}   f(\by_{z^k(\bc)} \mid \bth,  \bc),
\end{align*}
where $f(\by_{z^k(\bc)} \mid \bth,  \bc)$ is the marginal cluster likelihood obtained by integrating over all possible cluster parameters. Crucially, the dimension of the parameter space is fixed and now we may proceed to iteratively sample the global parameters and observation assignments.

Following \cite{blei2011}, the observation assignments can be updated via Gibbs sampling by sequentially sampling each $c_i$ from its full conditional distribution. For observation $i$, temporarily remove $i$'s link by setting $c_i = i$ to obtain the modified partition $z(\bc_{-i})$.

For observation $i$, we consider the change in the partition resulting from a new assignment denoted by $\cinew$. We first define some notation to describe the various cluster structures. Let $O = \mathcal{T}_i^{z(\textbf{c})}$ denote the original cluster, that is the set of indices associated with observation $i$ under the cluster assignments $z(\bc)$. Define the moving set $M= \mathcal{T}_i^{z(\bc_{-i})}$ to be the set of indices moving clusters with observation $i$, $R = O \setminus M$ to be the remaining set which is the indices not linked with $i$ and the target set $T = \mathcal{T}_i^{z(\bc_{-i} \cup \cinew)}$ to be the set of indices associated with observation $i$ in the new configuration $z(\bc_{-i} \cup \cinew)$.

The full conditional for a new observation assignment is therefore proportional to
\[
\pi(\cinew \mid \bth, \bc_{-i},  \by, \eta) \propto \begin{cases} 
      \alpha & \cinew = i \\
      f(d_{ij}) & \cinew \in \mathcal{T}_i^{z(\textbf{c})} \setminus \{i\} \\
      \frac{f(\by_{M \cup T} \mid \bth,  \bc)}{f(\by_{M} \mid \bth,  \bc) f(\by_{T} \mid \bth,  \bc)} & \text{otherwise}.
   \end{cases}
\]

The full conditional is available in closed form because the marginal cluster likelihoods are analytically tractable. Global parameters $\bth$, if present, are updated in separate Gibbs or Metropolis-Hastings steps conditional on the current partition.

\subsection{Inference for non-conjugate models}

When $\pi(\brho_k \mid \bc)$ is not conjugate to $f(\by_{z^k(\bc)} \mid \bth, \brho_k, \bc)$, the marginal cluster likelihood $f(\by_{z^k(\bc)} \mid \bth,  \bc)$ is not available in closed form and the cluster parameters $\brho = (\brho_1,...,\brho_{|z(\bc)|})$ must be retained in the state and sampled explicitly. The central difficulty is that reassigning an observation's link can change the number of clusters $K$, thereby creating or destroying cluster parameters and altering the dimension of the parameter space. This is a trans-dimensional inference problem, which we address using reversible jump Markov chain Monte Carlo (RJMCMC) \citep{green1995}.

We briefly review the RJMCMC framework. Consider two nested models $\mathcal{M}$ and $\mathcal{M}'$ with parameter vectors of possibly different dimensions. A move from $\mathcal{M}$ to $\mathcal{M}'$ proceeds by augmenting the current parameter vector with auxiliary random variables $\mathbf{u} \sim q(\mathbf{u})$ such that a deterministic bijection $g$ maps the augmented space to the proposed parameter space. Let $(\bth, \brho, \bc)$ and $(\bth, \brho', \bc')$ denote the current and proposed state of the Markov chain respectively, the acceptance probability for the proposed move is
$$
p_\text{accept} = \min\left(1, \frac{\pi(\bth, \brho', \bc' \mid \by, \eta)}{\pi(\bth, \brho, \bc \mid \by, \eta)} \cdot \frac{r(\mathcal{M}' \to \mathcal{M})}{r(\mathcal{M} \to \mathcal{M}')} \cdot \frac{q(\textbf{u}')}{q(\mathbf{u})} \cdot \left|\det \frac{\partial g (\boldsymbol{\rho}', \mathbf{u}')}{\partial (\boldsymbol{\rho}, \mathbf{u})}\right|\right),
$$
where $r(\mathcal{M} \to \mathcal{M}')$ is the probability of proposing a move from $\mathcal{M}$ to $\mathcal{M}'$, and the final term is the Jacobian determinant of the bijection $g$.

\subsection{RJMCMC for ddCRP models}

To update observation $i$'s assignment, we recall the moving set $M$, defined as the connected component containing $i$ when $i$'s outgoing link is temporarily removed. We propose a new configuration $\mathcal{M}'$ by sampling a new assignment $c_i^*$ for observation $i$. Two proposal strategies are

\paragraph*{Uniform}
We sample $j^* \sim \text{Uniform}(\{1,...,n\})$ and set $\cinew = j^*$. Since the probability of choosing any specific link is $1/n$ in both directions, the proposal distributions are symmetric:
\[
r(\mathcal{M} \to \mathcal{M}') = r(\mathcal{M}' \to \mathcal{M}) = 1,
\]
consequently the proposal ratio in the acceptance probability simplifies to 1.

\paragraph*{Prior proposal}
We sample a candidate link $j^*$ proportional to the ddCRP prior, i.e., with probability $\Pr(c_i=j^* \mid D, \alpha, f)$. In this case, $r(\mathcal{M} \to \mathcal{M}') = \Pr(c_i = j \mid D, \alpha, f)$ and the reverse move probability is $r(\mathcal{M}' \to \mathcal{M}) = \Pr(c_i = j^* \mid D, \alpha, f)$, where j is the current link. The proposal ratio becomes
\[
\frac{r(\mathcal{M}' \to \mathcal{M})}{r(\mathcal{M} \to \mathcal{M}')} = \frac{\Pr(c_i = j \mid D, \alpha, f)}{\Pr(c_i = j^* \mid D, \alpha, f)}.
\]
Note that this ratio is the inverse of the prior ratio found in the target term $\pi(\bth, \brho', \bc' \mid \by, \eta)/\pi(\bth, \brho, \bc \mid \by, \eta)$. Therefore, when using a prior proposal, the prior terms cancel out, and the acceptance probability depends primarily on the likelihood ratio.

When the proposed link $\cinew$ alters the number of connected components in the graph, we employ the reversible jump framework to traverse spaces of differing dimensionality. Let $K=|z(\bc)|$ denote the number of clusters under the current configuration, the proposed configuration may result in either an increase, decrease or the same number of clusters.

A birth move occurs when the new link $\cinew$ creates a cycle within the moving set $M$, isolating it as a new cluster. The state space expands from $\brho \in \Theta_K$ to $\brho'\in \Theta_{K+1}$. To match dimensions, we introduce auxiliary variables $\textbf{u} \sim q(\textbf{u})$ and set $(\brho', \textbf{u}') = g(\brho, \textbf{u})$.

The bijection $g$ is constructed component-wise. Writing $\brho_{\text{new}} = (\theta^{(1)}, \ldots, \theta^{(d)})$ to denote the $d$ scalar elements of the new cluster parameter vector, we draw independent auxiliary variables $u^{(p)} \sim q^{(p)}$ for $p = 1, \ldots, d$ and set $\theta_{\text{new}}^{(p)} = u^{(p)}$. Because the $d$ component bijections are independent, the Jacobian of the full mapping is block-diagonal and its determinant factorises as
$$
|\det J| = \prod_{p=1}^d |\det J^{(p)}|.
$$
For components proposed directly on their natural scale, the bijection is a variable relabelling with $|\det J^{(p)}| = 1$. When a component is instead proposed on a transformed scale, the change of variables introduces a non-trivial Jacobian factor; this arises for the log-normal proposal and is described in Section~\ref{sec:proposals}. Invertibility of $g$ is immediate, in the corresponding death move, the parameters of the cluster being destroyed serve as the reverse auxiliary variables $\mathbf{u}'$, and $g^{-1}$ is the same variable relabelling applied in the opposite direction.

The reverse move, which we define as a death move, occurs when the proposed link $\cinew$ connects the moving set $M$ (which was previously an isolated cluster) to a different component, merging two clusters into one. In this case the number of clusters reduces by one and set $(\brho, \textbf{u}) = g^{-1}(\brho',\textbf{u}')$. 

Finally, the case where the number of clusters do not change is defined as a fixed-dimensional move, which may arise in one of two cases. The first is if the proposed link $\cinew$ is already in the existing cluster, hence the moving set is equal to the target set and we have $M=T$. Assuming a symmetric proposal distribution for the configuration, in this case the acceptance probability depends only on the posterior ratio and is given by
$$
p_\text{accept} = \min\left(1, \frac{\pi(\bth, \brho', \bc' \mid \by, \eta)}{\pi(\bth, \brho, \bc \mid \by, \eta)} \right).
$$
Alternatively, we have the case where the partition changes but the number of clusters remains the same, which occurs when there are observations remaining in the original cluster. We may retain the original parameters or optionally propose new parameters designed to target regions of high posterior density. If the move also involves perturbing the cluster parameters (e.g., to explore the local mode of the new configuration), the acceptance ratio must be augmented to include the parameter proposal ratio and the likelihood terms for the new parameters.

\subsection{Proposal strategies} \label{sec:proposals}

Although detailed balance guarantees the theoretical validity of these moves, the practical performance of the algorithm depends on the proposal distribution. Proposals that place negligible mass near the posterior mode will yield low acceptance rates regardless of the asymptotic guarantees, and in trans-dimensional settings this can prevent the sampler from visiting configurations with different numbers of clusters altogether. We therefore consider a range of data-driven strategies for both birth and death moves and for fixed-dimensional parameter updates.

\subsubsection{Birth and death proposals}

For birth moves, we must propose new cluster parameters $\brho_{\text{new}}$ for the newly created cluster containing the moving set $M$. The proposal density $q(\brho_{\text{new}})$ appears in the acceptance ratio and directly influences mixing efficiency.

\paragraph{Independence proposal}
A straightforward approach is to use an independence sampler, where the new parameters are drawn from a fixed distribution $Q$ that does not depend on the current state:
\[
\brho_{\text{new}} \sim Q.
\]
The simplest choice for $Q$ is the prior distribution itself, where the proposal density is $q(\brho_{\text{new}}) = \pi(\brho_{\text{new}} \mid \bc)$. This approach is parameter-free and leads to convenient partial cancellations in the acceptance ratio. More generally, $Q$ can be a custom fixed distribution, with density $q(\brho_{\text{new}}) = f_Q(\brho_{\text{new}})$, tailored via domain knowledge or preliminary analysis to target specific regions of the parameter space.  In both cases the bijection sets $\theta^{(p)}_{\text{new}} = u^{(p)}$ for each scalar component, so $|\det J| = 1$ and the Jacobian term does not contribute to the acceptance probability.

However, regardless of the specific choice of $Q$, independence samplers do not use the information in the moving set $M$. When the prior is diffuse relative to the likelihood, proposed values will frequently fall far from regions of high posterior support, yielding low acceptance rates that are difficult to improve by tuning alone. This motivates data-driven proposals centred on summaries of $\by_M$, which we describe next.

\paragraph{Moment matching}
To exploit information in the data while maintaining computational efficiency, we employ moment-matching strategies that adaptively centre the proposal distribution around summary statistics of $\by_{M}$. Let $\hat{\brho}_{M}$ denote a method-of-moments estimator computed from the observations in the moving set. We consider three variants:

\textit{Normal moment matching.} For parameters with support on $\mathbb{R}$, we use a normal proposal:
\[
\brho_{\text{new}} \sim \text{Normal}(\hat{\brho}_{M}, \sigma^2),
\]
where $\sigma > 0$ is a tuning parameter controlling the dispersion around the moment estimate. The proposal is centred at a data-driven location but allows exploration via the specified variance. The bijection sets $\theta^{(p)}_{\text{new}} = u^{(p)}$ directly, so $|\det J^{(p)}| = 1$ for each component.

\textit{Inverse gamma moment matching.} For scale-type parameters, we fit an inverse gamma distribution via method of moments. Given data $\by_{M}$ with sample mean $\bar{y}$ and sample variance $s^2$, the method-of-moments estimators are
\[
\hat{\alpha} = 2 + \frac{\bar{y}^2}{s^2}, \quad \hat{\beta} = \bar{y}(\hat{\alpha} - 1),
\]
yielding the proposal $\brho_{\text{new}} \sim \text{InverseGamma}(\hat{\alpha}, \hat{\beta})$, provided $\hat{\alpha} > 2$ and $\hat{\beta} > 0$. If moment matching fails due to insufficient data or numerical issues, we can simply fall back to the prior proposal. As with the normal proposal, the bijection is a direct assignment with $|\det J^{(p)}| = 1$ for each component.

\textit{Log-normal moment matching.} For strictly positive parameters, the auxiliary variable for each scalar component is $u^{(p)} = \log \theta^{(p)}_{\text{new}}$, drawn as
\[
u^{(p)} \sim \text{Normal}(\log \hat{\brho}^{(p)}_{M}, \sigma^2),
\]
where $\hat{\brho}^{(p)}_{M}$ is the method-of-moments estimate for component $p$ computed from $\by_M$. The bijection $\theta^{(p)}_{\text{new}} = \exp(u^{(p)})$ maps the log-scale auxiliary to the original scale, introducing a Jacobian factor $|\partial \theta^{(p)}_{\text{new}} / \partial u^{(p)}| = \theta^{(p)}_{\text{new}}$ per component that must be included in the acceptance ratio. This proposal is particularly effective for parameters spanning multiple orders of magnitude, as it proposes symmetrically on the log scale while respecting the positivity constraint.

In all moment-matching proposals, if the moving set $M$ contains too few observations to reliably estimate moments (typically fewer than 2--3 observations), we default to the prior proposal to avoid unstable estimates.

\subsubsection{Fixed-dimensional updates}

When the proposed link $\cinew$ connects the moving set $M$ to a different existing cluster without changing the total number of clusters $K$, we have a fixed-dimensional move. In this case, the moving set $M$ transfers from its current cluster ($R$) to the target cluster ($T$). The cluster parameters may either remain unchanged or be updated to reflect the new configuration.

\paragraph*{No parameter update}
The simplest strategy maintains the existing cluster parameters:
\[
\brho_{R}' = \brho_{R}, \quad \brho_{T}' = \brho_{T},
\]
where $\brho_{R}'$ and $\brho_{T}'$ denote the parameters for the remaining and target clusters, respectively, after the move. Since the proposal is deterministic, the proposal ratio is unity and the acceptance probability depends solely on the posterior ratio:
\[
p_\text{accept} = \min\left(1, \frac{\pi(\bth, \brho', \bc' \mid \by, \eta)}{\pi(\bth, \brho, \bc \mid \by, \eta)}\right).
\]
This approach is computationally inexpensive but may result in poor likelihood values if the parameters are mismatched to the new cluster composition.

\paragraph*{Resampling from posterior}
To improve acceptance rates, we may resample cluster parameters from approximate conditional posteriors given the new cluster memberships. For each affected cluster $k \in \{R, T\}$, let $\by_k'$ denote the observations assigned to cluster $k$ in the proposed configuration. We then sample new parameters using moment-matched distributions as described previously. The proposal density for the forward move is $q(\brho_{R}', \brho_{T}' \mid \bc')$, and the reverse density is $q(\brho_{R}, \brho_{T} \mid \bc)$, yielding a Hastings ratio
\[
\frac{q(\brho_{R}, \brho_{T} \mid \bc)}{q(\brho_{R}', \brho_{T}' \mid \bc')}.
\]
Since the new parameters are proposed directly on their natural scale, the bijection for each scalar component is a direct assignment with $|\det J^{(p)}| = 1$, so no additional Jacobian term appears in the acceptance ratio. This strategy adapts the parameters to the new cluster composition, often improving acceptance rates at the cost of additional sampling operations if the proposal distribution is suitably chosen.

\subsection{Inference for ddCRP hyperparameters}
\label{sec:hyperparams}

The ddCRP prior with an exponential distance decay function introduces two hyperparameters: the concentration parameter $\alpha > 0$, which governs the propensity for self-links and hence the prior number of clusters, and the decay scale $s > 0$, which determines how rapidly co-clustering probability falls off with distance. In many applications these are fixed by the analyst, for instance by setting $s$ to a value that reflects the expected spatial or covariate scale of the clustering structure. It is, however, possible to infer $\alpha$ and $s$ within the MCMC algorithm by treating them as additional unknowns with prior distributions. We describe exact and approximate updates for each parameter in turn, before discussing the practical difficulties that arise when both are inferred jointly.

\subsubsection*{Inference for $\alpha$}

To infer the ddCRP concentration parameter $\alpha$, we may sample from its conditional posterior. Given a Gamma prior $\alpha \sim \text{Gamma}(a_\alpha, b_\alpha)$, the conditional posterior is given by:
\[
p(\alpha \mid \bc, s) \propto \alpha^{a_{\alpha} +n_{\text{self}}-1} \exp(-b_\alpha \alpha) \prod_{i=1}^n \frac{1}{\alpha + R_i(s)}.
\]
The product term $\prod_{i=1}^n \frac{1}{\alpha + R_i(s)}$ arising from the ddCRP normalisation prevents a standard conjugate update. To bypass this intractable denominator, we exploit a data-augmentation strategy, recognising that each fractional term is the normalising constant of an Exponential distribution.

\begin{proposition}
    Let $\bc=(c_1,...,c_n)$ be the observation assignments that contain exactly $n_{\text{self}}$ self-links, and let $R_i(s) = \sum_{j \neq i} f(d_{ij})$ by the total decay weight for observation $i$. The exact posterior distribution $p(\alpha \mid \bc, s)$ can be targeted by introducing conditionally independent auxiliary variables $V_1,\ldots,V_n$ with the following exact Gibbs updates:
    \begin{enumerate}
        \item $V_i \mid \alpha, \bc, s \sim \text{Exponential}(\alpha + R_i(s))$ for $i=1,\ldots,n$
        \item $\alpha \mid V, \bc, s \sim \text{Gamma}(a_\alpha + n_{\text{self}}, b_\alpha + \sum_{i=1}^n V_i)$ 
    \end{enumerate}
\end{proposition}
\begin{proof}
The full derivation, verifying that marginalising over the auxiliary variables $V$ strictly recovers the original target posterior $p(\alpha \mid \bc, s)$, is provided in Appendix 1.
\end{proof}

\subsubsection*{Inference for $s$ via Metropolis--Hastings}

Since $R_i(s)$ is nonlinear in $s$, no known conjugate update for the decay scale exists. We instead use a log-normal random-walk Metropolis--Hastings step. Given the current value $s$, we propose
$$
s' = s \exp(\varepsilon), \qquad \varepsilon \sim \mathrm{Normal}(0, \sigma_s^2),
$$
where $\sigma_s^2$ is a tuning parameter. The Jacobian of the log-scale proposal introduces an additional factor of $s'/s$ in the acceptance ratio. When $\alpha$ is being inferred simultaneously, the auxiliary variables $V_i$ sampled in the $\alpha$ update can be reused. Specifically, the log acceptance ratio for the augmented target is
$$
\log r = a_s \log\frac{s'}{s} - (b_s + D_{\mathrm{sum}})(s' - s) - \sum_{i=1}^n V_i\!\left[R_i(s') - R_i(s)\right],
$$
where $D_{\mathrm{sum}} = \sum_{i:\,c_i \neq i} d_{i,c_i}$ is the sum of distances over non-self links and $b_s$ is the rate of the Gamma prior $s \sim \mathrm{Gamma}(a_s, b_s)$. When $\alpha$ is held fixed and only $s$ is inferred, the auxiliary variables are not available and the log acceptance ratio reverts to
$$
\log r = a_s \log\frac{s'}{s} - (b_s + D_{\mathrm{sum}})(s' - s) - \sum_{i=1}^n \log\frac{\alpha + R_i(s')}{\alpha + R_i(s)}.
$$
In either case, the proposed $s'$ is accepted with probability $\min(1, e^{\log r})$.

\subsubsection*{Practical considerations}

Although the updates above are theoretically exact, jointly inferring $\alpha$ and $s$ is subject to a weak identifiability problem that can seriously impede mixing in practice. Both parameters enter the ddCRP likelihood exclusively through the normalising constants $\alpha + R_i(s)$, where $R_i(s) = \sum_{j \neq i} \exp(-s\,d_{ij})$. When the distance scale $s$ is small relative to the spread of pairwise distances, the decay surface $e^{-sd}$ is nearly flat and $R_i(s)$ varies little across observations. In this scenario, an increase in $s$ and a compensating increase in $\alpha$ can yield almost identical normalised assignment probabilities, so the joint posterior over $(\alpha, s)$ develops a ridge of near-equal density along which the individual parameters are not separately identified by the data. This is a manifestation of parameter confounding: $\alpha$ and $s$ are weakly identifiable when the clustering signal is diffuse, in the sense that multiple $(\alpha, s)$ pairs are consistent with the observed partition up to the precision of the likelihood.

The practical consequence is that the Metropolis--Hastings step for $s$ will mix slowly along this ridge regardless of the choice of $\sigma_s^2$, since proposals that move across the ridge are penalised by the likelihood while proposals that remain on it are essentially undirected. This difficulty is not primarily a tuning problem and cannot be resolved by adjusting the proposal variance alone. Mixing improves only when the data strongly constrain the distance scale, that is when cluster boundaries are well-defined and the partition is sensitive to the value of $s$ across the observed distance range.

In practice, fixing $s$ at a value informed by the expected covariate scale of the clustering structure and inferring only $\alpha$ is often the more stable approach. The data-augmented Gibbs step for $\alpha$ is exact and requires no tuning, and the identifiability issue does not arise for $\alpha$ alone since its conditional posterior given $s$ and $\bc$ is a standard Gamma distribution. When $s$ is inferred, we recommend verifying that the marginal posterior for $s$ differs materially from the prior before treating the inferred value as informative, since a posterior that closely tracks the prior is consistent with weak identifiability rather than genuine data-driven constraint.

\section{Results}

We evaluate the proposed inference methodology through two analyses. We first present a simulation study using Poisson data, which admits a conjugate formulation and allows sampler correctness to be verified by direct comparison with the marginalised Gibbs baseline. We then apply the methodology to the Old Faithful geyser dataset, which requires a non-conjugate gamma likelihood and illustrates the method in a practically relevant continuous-data setting.

\subsection{Simulation study}

In this section we present a simulation study where the data are clustered and within each group the data are generated from a Poisson distribution with cluster specific rates $\lambda_k$ for $k=1,..., K^*$ where $K^*$ is the maximum number simulated of clusters.

The Poisson observation model provides a convenient setting in which to verify the correctness of the proposed RJMCMC samplers. Since the Gamma distribution is conjugate to the Poisson likelihood, the cluster-specific rate parameters $\lambda_k$ can be marginalised out analytically, yielding an exact collapsed Gibbs sampler that targets the same posterior as the non-conjugate RJMCMC approaches. We exploit this by running a marginalised Gibbs baseline alongside a suite of RJMCMC samplers on the same simulated data, and verifying that the resulting posterior distributions are in agreement.

\subsubsection{Data generation and model description}

We simulate $n = 150$ observations from a Poisson clustering model with $K^* = 3$ true clusters of equal size. Covariates are drawn from a Gaussian mixture with centres $\mu_1 = -3$, $\mu_2 = 0$, $\mu_3 = 3$ and common standard deviation $\sigma = 1.5$, so that the clusters partially overlap in covariate space. Specifically, we set
$$
x_i \mid z_i = k \;\sim\; \mathrm{Normal}(\mu_k,\, \sigma^2),
$$
where $z_i \in \{1, 2, 3\}$ is the true cluster label. The ddCRP distance matrix is constructed from absolute pairwise covariate differences, $D_{ij} = |x_i - x_j|$, with an exponential decay kernel $f(d) = \exp(-s\,d)$ and fixed $s = 0.5$. We focus on an overlapping scenario in which the true cluster rates are $\blam = (1, 4, 7)$. The moderate separation between rates, combined with the covariate overlap, makes this a nontrivial inference setting. Observations are drawn as
$$
y_i \mid z_i = k \;\sim\; \mathrm{Poisson}(\lambda_k).
$$
Figure~\ref{fig:pois_data} shows the distribution of counts within each true cluster and the scatter of covariate against count coloured by true cluster membership.

\begin{figure}[htbp]
\centering
\includegraphics[width=\textwidth]{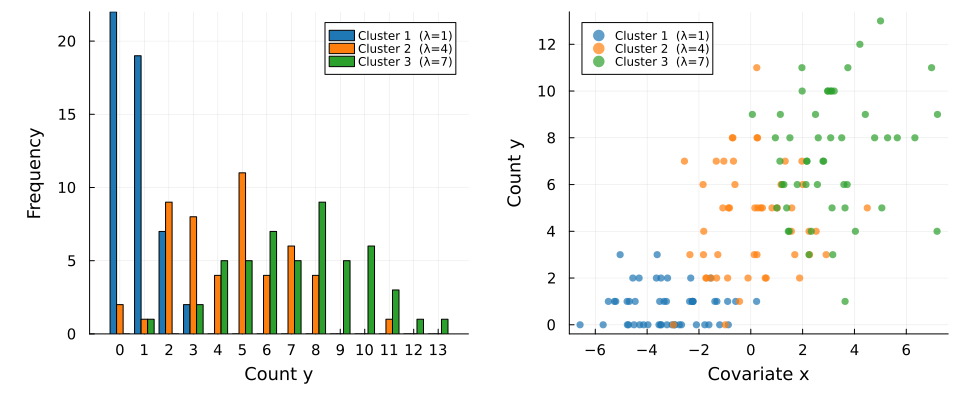}
\caption{Simulated Poisson data for the overlapping scenario ($n = 150$, $K^* = 3$,
  $\blam = (1, 4, 7)$). Left: histogram of counts stratified by true cluster membership.
  Right: scatter of covariate $x$ against count $y$, coloured by true cluster.}
\label{fig:pois_data}
\end{figure}

Recall that $\by_{z^k(\bc)}$ denotes the vector of observations in cluster $k$ and $\by$ is the complete vector of observations. Conditional on the cluster assignments, which are induced by the ddCRP prior and observation assignment vector $\bc$, the joint probability mass function of the data is
$$
f(\by \mid \blam) = \prod_{k=1}^{|z(\bc)|} f(\by_{z^k(\bc)} \mid \lambda_k),
$$
where $f(\by_{z^k(\bc)} \mid \lambda_k)$ is the joint probability mass function of observations in cluster $k$ and explicitly is given by
$$
f(\by_{z^k(\bc)} \mid \lambda_k) = \prod_{i \in z^k(\bc)} \frac{\lambda_k^{y_i}e^{-\lambda_k}}{y_i!},
$$
which follows from the fact that observations within a cluster are independent of one another.

By assigning independent conjugate priors $\lambda_k \sim \text{Gamma}(a,b)$ and the ddCRP prior, the posterior is
\begin{align*}
    \pi(\blam, \bc \mid \by) & \propto f(\by \mid \blam) \pi(\blam) \pi(\bc \mid \eta) \\
    & = \pi(\bc \mid \eta) \prod_{k=1}^{|z(\bc)|} \left( \prod_{i \in z^k(\bc)} \frac{\lambda_k^{y_i} e^{- \lambda_k}}{y_i!} \lambda_k^{a-1} e^{-b \lambda_k}  \right) \\
    &= \pi(\bc \mid \eta) \prod_{k=1}^{|z(\bc)|} \left(  \lambda_k^{S_k + a-1} e^{-\lambda_k (n_k  + b)} \prod_{i \in z^k(\bc)} \frac{1}{y_i!}  \right),
\end{align*}
where $S_k$ and $n_k$ are the sum and number of observations in cluster $k$ respectively and $\eta = (\textbf{D}, \alpha, f, s)$ is the vector of ddCRP hyperparameters consisting of pairwise distances $\textbf{D}$, self link probability $\alpha$, distance decay function $f$ and the rate of decay $s$. We aim to explicitly compare the reversible jump methods introduced in this article with the standard conjugate implementation described in \cite{blei2011}, therefore for completeness we write the marginalised posterior:
\begin{align*}
    \pi(\bc \mid \by) &= \int \pi(\blam, \bc \mid \by) d\blam \\
    &= \pi(\bc \mid \eta) \prod_{k=1}^{|z(\bc)|} \left( \int  \lambda_k^{S_k + a-1} e^{-\lambda_k (n_k  + b)} d \lambda_k \prod_{i \in z^k(\bc)} \frac{1}{y_i!}  \right) \\
    &= \pi(\bc \mid \eta) \prod_{k=1}^{|z(\bc)|} \left( \frac{\Gamma(S_k + a)}{(n_k + b)^{S_k + a}} \prod_{i \in z^k(\bc)} \frac{1}{y_i!}  \right).
\end{align*}

\subsubsection{Inference}

We compare seven sampler configurations. The marginalised Gibbs baseline analytically integrates out $\lambda_k$ and serves as the reference for correctness verification. The remaining six configurations use the non-conjugate model, in which the cluster rates are maintained as explicit sampler parameters. Birth proposals are drawn from one of three families: the prior, a normal moment-matched proposal (NMM), and a log-normal moment-matched proposal LNMM). Each birth proposal is crossed with two fixed-dimensional update strategies: No Update, in which cluster rates are left unchanged during reassignment moves, and Resample, in which the affected cluster rates are resampled from a moment-matched approximate posterior after each reassignment. For each of the proposed inference strategies we assign independent uninformative Gamma(1, 0.1) priors to the rate parameters $\lambda_k$.

To select the birth proposal standard deviation $\sigma_b$ for the NMM and LNMM families, we conduct a grid search using shorter preliminary chains of 50{,}000 post-burn-in samples. For No Update configurations we sweep $\sigma_b \in \{0.05, 0.10, 0.25, 0.50, 1, 2\}$; for Resample configurations we search a two-dimensional grid $\sigma_b \times \sigma_r$ over the same set of values, where $\sigma_r$ is the standard deviation of the resampling proposal. The tuning criterion is the expected squared jumping distance of $K$ (ESJD$(K)$), defined as the mean of $\{(K^{(t+1)} - K^{(t)})^2\}$ over thinned post-burn-in samples, which rewards chains that move frequently between distinct cluster configurations \citep{esjd-gelman-paper}. The best-performing parameter combination from each group is then used in the final full-length runs of 100{,}000 post-burn-in iterations.

\subsubsection{Results}

Table~\ref{tab:pois_overlapping} summarises the performance of each method on the simulated data in terms of the posterior mean and mode of $K$, $P(K{=}K^*)$, and ESS$(K)$. All methods recover the correct mode $K^* = 3$ and produce comparable posterior estimates of $\bar{K}$, and $P(K{=}K^*)$. The marginalised Gibbs sampler achieves the highest ESS$(K)$ by a substantial margin, as expected given its collapsed representation; however, this advantage comes at the cost of restricting inference to the conjugate Poisson model. Among the RJMCMC samplers, the prior, NMM, and LNMM configurations achieve broadly comparable mixing efficiency, with No Update configurations tending to yield slightly higher ESS$(K)$ than their Resample counterparts for this model.

\begin{table}[htbp]
\centering
\caption{Comparison of sampler configurations for the Poisson simulated data ($\lambda = (1, 4, 7)$, $n = 150$, $K^* = 3$). Results are computed over 20{,}000 thinned post-burn-in samples.}
\label{tab:pois_overlapping}

\begin{tabular}{l c c r r r r}
  \toprule
  Birth & Fixed-dim & $\bar{K}$ & Mode $K$ & $P(K{=}K^*)$ & ESS$(K)$ & Birth acc (\%) \\
  \midrule
  Gibbs & none & 4.44 & 3 & 0.246 & 16030 & — \\
  \midrule
  Prior & none & 4.50 & 3 & 0.239 & 3720 & 19.5 \\
  Prior & resample & 4.50 & 3 & 0.245 & 3655 & 19.4 \\
  \midrule
  NMM & none & 4.46 & 3 & 0.241 & 2648 & 15.7 \\
  NMM & resample & 4.40 & 3 & 0.245 & 3166 & 15.8 \\
  \midrule
  LNMM & none & 4.45 & 3 & 0.240 & 3702 & 16.2 \\
  LNMM & resample & 4.44 & 3 & 0.244 & 2967 & 17.4 \\
  \bottomrule
\end{tabular}

\end{table}

Figure~\ref{fig:pois_traces} displays trace plots for the number of clusters $K$ and the ddCRP concentration parameter $\alpha$ over the post-burn-in chain. The Gibbs chain mixes rapidly across values of $K$, reflecting the efficiency of the collapsed representation. The RJMCMC chains show comparable mixing in $K$, with all methods making frequent transitions between cluster configurations. The $\alpha$ traces confirm that the concentration parameter is well-identified and that all methods converge to the same marginal posterior.

\begin{figure}[htbp]
\centering
\includegraphics[width=\textwidth]{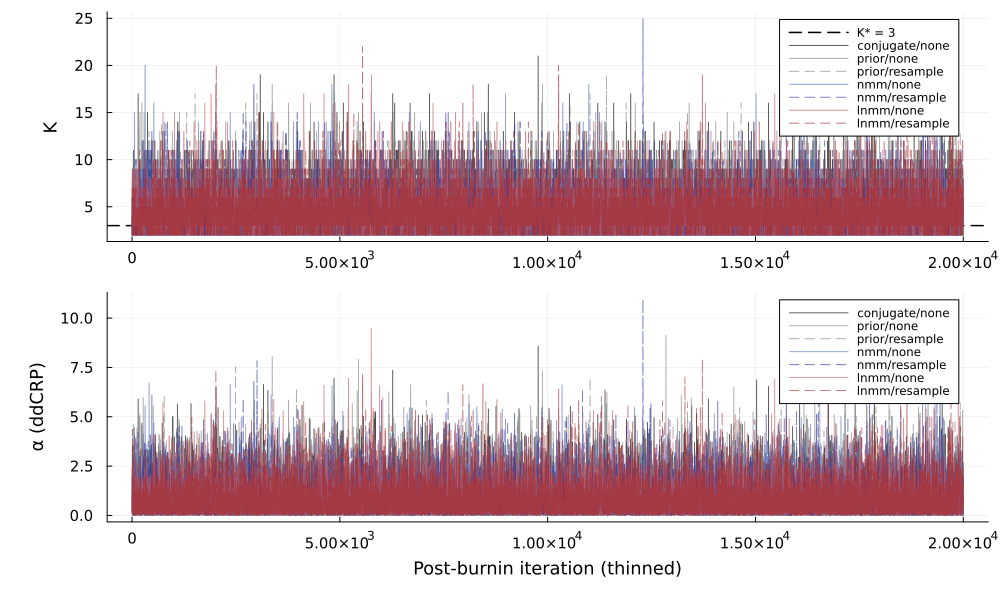}
\caption{Trace plots for the Poisson simulated data. Top: number of clusters $K$ over post-burn-in iterations (thinned); the dashed horizontal line marks the true value $K^* = 3$. Bottom: ddCRP concentration parameter $\alpha$.}
\label{fig:pois_traces}
\end{figure}

Figure~\ref{fig:pois_pk} shows the posterior distribution $P(K \mid \cdot)$ for each method. All samplers place the majority of posterior mass at $K = 3$ and the distributions are in close agreement across methods, providing strong evidence that the RJMCMC samplers correctly target the same posterior as the marginalised Gibbs baseline in this example.

\begin{figure}[htbp]
\centering
\includegraphics[width=0.75\textwidth]{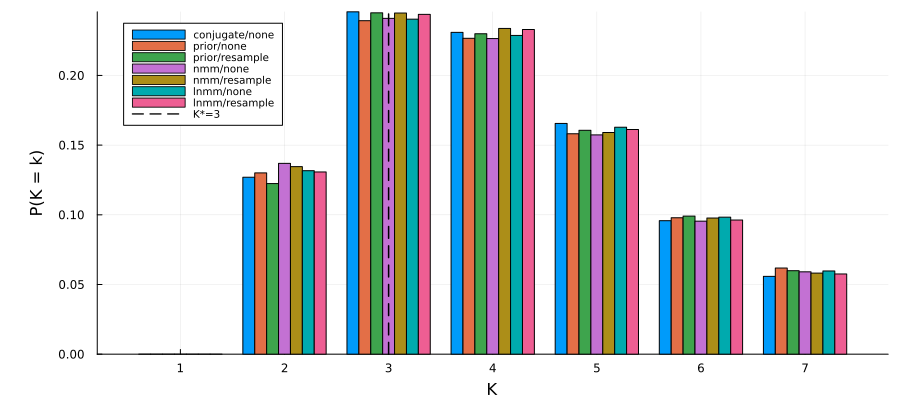}
\caption{Marginal posterior distribution $P(K \mid \bc, \by)$ for all sampler configurations. The dashed vertical line marks the true number of clusters $K^* = 3$.}
\label{fig:pois_pk}
\end{figure}

\subsection{Application to Old Faithful eruption data}

We apply the method to the classic Old Faithful geyser dataset, modelling eruption durations as the response and using waiting times between eruptions as covariate information in the ddCRP prior.

\subsubsection{Data description}

The Old Faithful dataset consists of $n = 272$ observations of eruptions of the Old Faithful geyser in Yellowstone National Park. For each eruption $i$, the data are $y_i > 0$, the eruption duration (in minutes) and $w_i > 0$, the waiting time until the next eruption (in minutes). The dataset exhibits well-known bimodal behaviour: short eruptions (around 2 minutes) are associated with short waiting times (around 55 minutes), while long eruptions (around 4.5 minutes) are associated with long waiting times (around 80 minutes). We construct the distance matrix based on the absolute difference in waiting times ($d_{ij} = |w_i - w_j|$), encouraging eruptions with similar waiting times to cluster together. The ddCRP prior uses exponential decay $f(d) = \exp(-s \cdot d)$. We fix the decay scale to $s = 0.2$ and infer the concentration parameter $\alpha$ within the MCMC algorithm using the data-augmented Gibbs step described in Section~\ref{sec:hyperparams}.

\subsubsection{Model specification and marginalisation}

Let $z_i(\bc) \in \{1, \ldots, |z(\bc)|\}$ denote the cluster index of observation $i$ under the partition induced by observation assignment vector $\bc$. We model eruption durations using a gamma likelihood, a flexible choice for positive continuous data. Specifically, we write
$$y_i \mid \bc, \balpha, \bbeta \sim \text{Gamma}(\alpha_{z_i(\bc)}, \beta_{z_i(\bc)}),$$
where $\alpha_k > 0$ and $\beta_k > 0$ are the cluster-specific shape and rate parameters for cluster $k$. Note that $\alpha_k$ and $\beta_k$ here are model parameters and are distinct from the ddCRP concentration $\alpha$. We assign independent gamma priors to both cluster parameters:
\begin{align*}
\alpha_k &\sim \text{Gamma}(a_\alpha, b_\alpha) \\
\beta_k &\sim \text{Gamma}(a_\beta, b_\beta)
\end{align*}
with hyperparameters $a_\alpha = 2$, $b_\alpha = 0.5$, $a_\beta = 2$, and $b_\beta = 0.5$, yielding weakly informative priors centred near the empirical mean eruption duration.

Because the gamma prior on the rate parameter $\beta_k$ is conjugate to the gamma likelihood, we can analytically marginalise out $\beta_k$ to obtain a marginal cluster likelihood conditional only on the shape parameter $\alpha_k$. Let $z^k(\bc)$ denote the set of indices for observations in cluster $k$, and $n_k = |z^k(\bc)|$ be the number of observations in that cluster. The integration over the prior distribution of $\beta_k$ proceeds as follows:
\begin{align*}
f(\by_{z^k(\bc)} \mid \alpha_k) &= \int_0^\infty \left( \prod_{i \in z^k(\bc)} \frac{\beta_k^{\alpha_k}}{\Gamma(\alpha_k)} y_i^{\alpha_k - 1} e^{-\beta_k y_i} \right) \frac{b_\beta^{a_\beta}}{\Gamma(a_\beta)} \beta_k^{a_\beta - 1} e^{-b_\beta \beta_k} d\beta_k \\
&= \left( \prod_{i \in z^k(\bc)} \frac{y_i^{\alpha_k - 1}}{\Gamma(\alpha_k)} \right) \frac{b_\beta^{a_\beta}}{\Gamma(a_\beta)} \int_0^\infty \beta_k^{n_k \alpha_k + a_\beta - 1} e^{-\beta_k \left(b_\beta + \sum_{i \in z^k(\bc)} y_i\right)} d\beta_k \\
&= \left( \prod_{i \in z^k(\bc)} \frac{y_i^{\alpha_k - 1}}{\Gamma(\alpha_k)} \right) \frac{b_\beta^{a_\beta}}{\Gamma(a_\beta)} \frac{\Gamma(n_k \alpha_k + a_\beta)}{\left(b_\beta + \sum_{i \in z^k(\bc)} y_i\right)^{n_k \alpha_k + a_\beta}}
\end{align*}

This marginalisation substantially simplifies the state space. Instead of proposing jointly for two non-conjugate parameters per cluster, the trans-dimensional problem is reduced such that we only need to infer and update a single non-conjugate parameter, $\alpha_k$, for each cluster. The target joint posterior distribution for the assignments and shape parameters is therefore given by:
$$ \pi(\bc, \balpha \mid \by, \eta) \propto \pi(\bc \mid \eta) \prod_{k=1}^{|z(\bc)|} f(\by_{z^k(\bc)} \mid \alpha_k) \pi(\alpha_k), $$
where $\eta$ encapsulates the ddCRP hyperparameters.

\subsubsection{Inference}

Since the prior on the shape parameter $\alpha_k$ is not conjugate to the resulting marginal likelihood, we employ a hybrid MCMC scheme. The state space consists of the assignments $\bc$ and the vector of shape parameters $\balpha = (\alpha_1, \ldots, \alpha_{|z(\bc)|})$.

For updates to the partition $\bc$, we use the RJMCMC framework. For birth moves, when a new cluster is formed, we must propose a new shape parameter $\alpha_{\text{new}}$. To target the region of high posterior density efficiently, we utilise an Inverse Gamma moment-matching proposal. For fixed-dimensional moves, where an observation reassignment transfers the moving set between existing clusters without altering the total number of clusters, we employ a posterior resampling strategy. We resample the affected shape parameters directly from moment-matched approximate posteriors evaluated on the updated cluster memberships, which prevents parameter-assignment mismatch and maintains robust mixing efficiency.

Conditional on a fixed partition $\bc$, the shape parameters $\alpha_k$ are updated using a Metropolis-Hastings step. Because the shape parameter is strictly positive, we employ a random walk on the logarithmic scale. We propose a new value $\alpha_k^*$ such that
$$ 
\log \alpha_k^* \sim \text{Normal}(\log \alpha_k, \sigma^2), 
$$
where $\sigma^2$ is a tuning parameter adjusted during burn-in to achieve an optimal acceptance rate. The acceptance probability for this move is:
$$ p_\text{accept} = \min\left(1, \frac{\pi(\alpha_k^*) f(\by_{z^k(\bc)} \mid \alpha_k^*)}{\pi(\alpha_k) f(\by_{z^k(\bc)} \mid \alpha_k)} \times \frac{\alpha_k^*}{\alpha_k} \right) $$
where the final term $\alpha_k^* / \alpha_k$ is the Jacobian determinant associated with the log-scale proposal. We assign an Exponential(0.01) prior to the ddCRP concentration parameter $\alpha$ and infer it within the MCMC algorithm using the data-augmented Gibbs step described in Section~\ref{sec:hyperparams}.

\subsubsection{Results}

The main RJMCMC sampler was run for 20,000 iterations (5,000 burn-in), Figure~\ref{fig:of_k} shows the posterior distribution of the number of clusters $K$ (left panel) and trace plots of $K$ and the log-posterior over the post-burn-in chain (right panel). The posterior concentrates strongly on $K = 2$, consistent with the well-documented bimodal behaviour of the geyser: eruptions cluster into a short-duration group (approximately 2 minutes) associated with short waiting times, and a long-duration group (approximately 4.5 minutes) associated with long waiting times. 

The trans-dimensional proposals achieved a 0.6\% acceptance rate for both birth and death moves, which is standard for RJMCMC algorithms navigating dimension-changing parameter spaces, as discussed in Section~\ref{sec:proposals}. Fixed-dimensional parameter updates mixed substantially more efficiently, with an acceptance rate of 33.1\%.

The Markov chain exhibited healthy mixing across the target posterior. The effective sample size (ESS) for the log-posterior was 4193.9, and the ESS for the number of clusters, $K$, was 325.8. For the observation-level shape parameters, we aggregated the ESS over all $n$ observations; the median ESS was 2777.9, with a worst-case minimum ESS of 71.3. This indicates that even the most slowly mixing individual parameters explored their conditional posteriors adequately.

\begin{figure}[htbp]
\centering
\begin{subfigure}[b]{0.35\textwidth}
  \includegraphics[width=\textwidth]{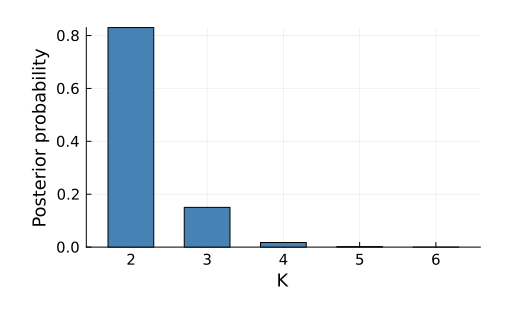}
  \caption{Posterior distribution of $K$}
\end{subfigure}\hfill
\begin{subfigure}[b]{0.62\textwidth}
  \includegraphics[width=\textwidth]{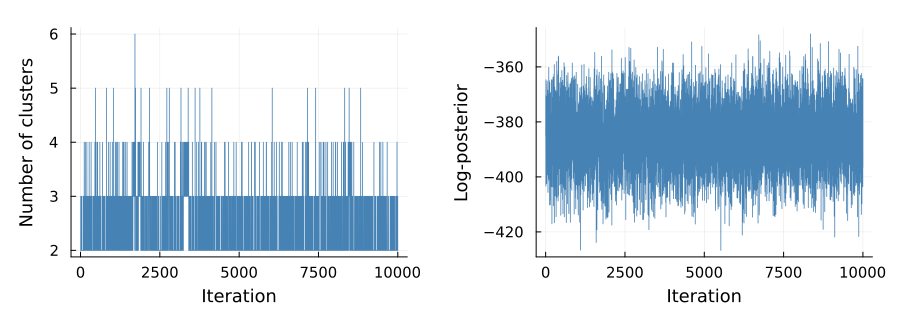}
  \caption{Trace plots (post burn-in).}
\end{subfigure}
\caption{Posterior inference on the number of clusters $K$ for the Old Faithful analysis. Left: posterior probability mass function of $K$, strongly concentrated at $K = 2$. Right: trace plots of $K$ (top) and log-posterior (bottom) over the 15000 post-burn-in iterations, showing good mixing with rapid transitions between the dominant $K = 2$ state and occasional visits to $K = 3$.}
\label{fig:of_k}
\end{figure}

\paragraph{MAP cluster assignments}
Figure~\ref{fig:of_scatter} shows the MAP cluster assignments, where each observation is coloured according to the MAP cluster assignment, obtained as the highest log-posterior sample among those with $K = K^*$ clusters, where $K^*$ denotes the posterior mode of $K$. For each observation $i$, we compute the posterior link probability $P(c_i = j \mid K = K^*)$. Arrows are drawn from each observation to its three most probable link targets, together with any additional targets whose posterior link probability exceeds $0.02$. 

The two clusters cleanly separate eruptions according to waiting time, with the boundary falling near $w \approx 65$ minutes. The MAP partition identifies a short-eruption cluster characterised by durations near 2 minutes and waiting times below 65 minutes, and a long-eruption cluster with durations near 3--5 minutes and waiting times above 65 minutes.

\begin{figure}[htbp]
\centering
\includegraphics[width=0.65\textwidth]{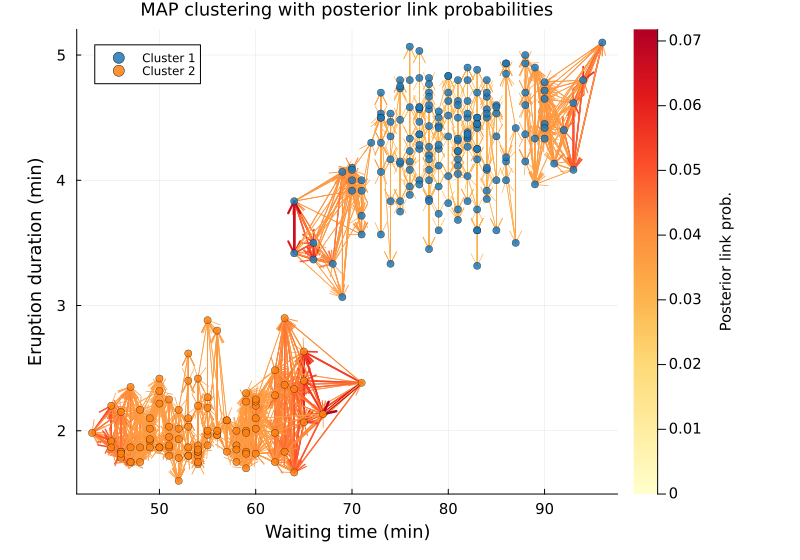}
\caption{MAP cluster assignments for the Old Faithful dataset ($n = 272$). Points are coloured by cluster membership under the maximum a posteriori partition. The two clusters correspond to short eruptions with short waiting times and long eruptions with long waiting times, with a clear separation near $w \approx 65$ minutes.}
\label{fig:of_scatter}
\end{figure}

\paragraph*{Posterior predictive check}
Figure~\ref{fig:of_ppd} shows the posterior predictive check. For each observation $i$, we plot the observed eruption duration alongside the posterior predictive mean and 95\% credible interval, obtained by sampling $\beta_k$ from its posterior conditional on the cluster assignment and shape parameter at each iteration, then drawing replicate eruption durations. The model captures the bimodal response distribution well: observations in both clusters fall within their respective predictive intervals for the great majority of cases, and the posterior predictive means track the two clusters closely. There is no evidence of systematic misfit.

\begin{figure}[htbp]
\centering
\includegraphics[width=0.88\textwidth]{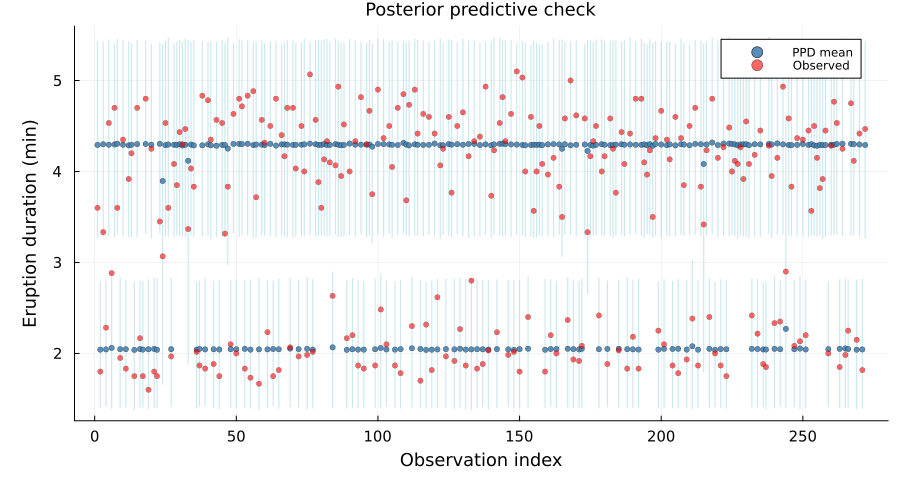}
\caption{Posterior predictive check for the Old Faithful analysis. Red points are observed eruption durations; blue points and vertical bars show the posterior predictive mean and 95\% credible interval at each observation index. Observations are not reordered, so the alternating structure reflects the natural ordering of the data. The model adequately captures the bimodal response distribution with tight predictive intervals throughout both clusters.}\label{fig:of_ppd}
\end{figure}

\section{Discussion}

We have developed an RJMCMC framework for posterior inference in ddCRP models with non-conjugate cluster likelihoods, extending the Gibbs sampler of \citet{blei2011} to settings where cluster-specific parameters cannot be analytically marginalised. The central methodological challenge is that each observation reassignment may alter the number of clusters and hence the dimension of the parameter space; we address this through the RJMCMC framework of \citet{green1995}, introducing data-driven moment-matching proposals that centre new cluster parameters on summary statistics of the observations in the moving set. For birth and death moves we compare Normal, InverseGamma, and log-normal moment-matching proposals; InverseGamma proposals are consistently preferred for scale parameters, since their structural alignment with InverseGamma priors avoids the systematic under-clustering bias exhibited by log-normal proposals in the overlapping scenario. For fixed-dimensional moves, resampling all affected cluster parameters from approximate posteriors evaluated on the updated cluster memberships yields the largest gains in mixing efficiency, achieving effective sample sizes more than three times larger than no-update strategies in the simulation study. We derive an exact data-augmented Gibbs step for the ddCRP concentration parameter $\alpha$ by introducing exponential auxiliary variables whose densities cancel the normalising constants in the conditional posterior, yielding an exact Gamma full conditional at each iteration. The simulation study confirms sampler correctness by comparing marginalised Gibbs and RJMCMC outputs in the conjugate Poisson setting, and the Old Faithful application demonstrates that the methodology recovers accurate clustering structure in a non-conjugate continuous-data setting. The framework also extends to prediction at unobserved covariate locations, where outcomes for new points with known covariates but no observed responses are required. A notable complication is that the ddCRP lacks marginal invariance, so that posterior samples from the fitted chain cannot be combined directly with an augmented distance matrix; we derive the exact posterior predictive distribution and a computationally tractable sequential approximation that reuses samples from the observed-data chain in Appendix 2.

Several limitations of the current methodology merit discussion. The most fundamental is that birth and death acceptance rates of approximately 1--2\% are an inherent feature of trans-dimensional moves in non-conjugate settings: even a proposal well-centred at the local posterior mode will yield low acceptance rates due to the unfavourable geometry of the joint parameter space, and this rate cannot be substantially improved by tuning alone. The decay function $f$ is fixed to the exponential form throughout; the scale $s$ is inferred but the functional form itself is not, and when the true distance-decay relationship is non-exponential this misspecification may introduce bias into both the partition inference and the posterior on $s$. The simulation study is restricted to Poisson and gamma mixture models, and behaviour on other non-conjugate likelihoods, such as Weibull, Student-$t$, skew-normal, or negative binomial with unknown dispersion, is not characterised; the proposal recommendations in Section~\ref{sec:proposals} should therefore be applied cautiously in settings that differ materially from those studied here. Finally, each MCMC iteration requires $O(n)$ observation assignment updates; for the datasets considered here this is not prohibitive, but for substantially larger collections of observations the per-iteration cost would become problematic without approximation.

Several natural directions follow from this work. Gradient-based proposals, such as Metropolis-adjusted Langevin steps, could in principle be used within birth moves for smooth likelihoods to improve the alignment of proposed parameters with the local posterior, exploiting automatic differentiation of the cluster likelihood with respect to cluster parameters. For the distance-decay relationship, a non-parametric prior, such as a monotone Gaussian process, could replace the parametric exponential, allowing the distance-decay relationship to be inferred from data rather than specified \textit{a priori}. Parallel tempering, running multiple chains at different inverse temperatures and exchanging states periodically, is a broadly applicable strategy for improving mixing over $K$ and would be straightforward to integrate with the existing RJMCMC framework. Extensions to multivariate responses, hierarchical models with shared cluster structure across multiple populations, and longitudinal settings where cluster membership evolves over time are all, in principle, accommodated by the general dimension-matching argument developed here, since the RJMCMC acceptance probability depends only on the proposal density and its Jacobian, not on the specific form of the cluster likelihood.

The methodology described in this article is implemented in an open-source Julia package designed to provide a flexible, general-purpose framework for distance-dependent Chinese Restaurant Process modelling. \texttt{DistanceDependentCRP.jl} provides both conjugate samplers (marginalised Gibbs) and non-conjugate samplers (RJMCMC) for all included models and implements the full suite of proposal strategies described in Section~\ref{sec:proposals}.

\if1\anon
{
 \section*{Acknowledgments}
    This work was supported by a UKRI Future Leaders Fellowship under grant [MR/X034992/1] and the University of Nottingham. The computations described in this paper were performed using the University of Birmingham's BlueBEAR HPC service, which provides a High Performance Computing service to the University's research community. See \url{http://www.birmingham.ac.uk/bear} for more details.
} \fi

\section*{Disclosure Statement}
The authors report there are no competing interests to declare.

\section*{Data availability}
All code and data used in this analysis can be found in a repository at \url{https://github.com/jmarsh96/non-conjugate-ddcrp}. The original Old Faithful data set can be found in the \textit{datasets} R package.

\bibliography{references.bib}

\end{document}